\documentclass[sigconf]{acmart}

\usepackage[utf8]{inputenc}

\usepackage{amsmath,amssymb}
\usepackage{amsthm}
\usepackage{mathrsfs}
\usepackage{stmaryrd}
\usepackage{enumerate}
\usepackage[algoruled,vlined,english,linesnumbered]{algorithm2e}
\usepackage{comment}
\usepackage{multirow}
\usepackage{xspace}

\usepackage{array}
\newcommand{\PreserveBackslash}[1]{\let\temp=\\#1\let\\=\temp}
\newcolumntype{C}[1]{>{\PreserveBackslash\centering}p{#1}}

\usepackage{tikz}
\usetikzlibrary{calc, matrix, backgrounds,shapes.geometric}

\usepackage{mathtools}
\mathtoolsset{showonlyrefs,showmanualtags}

\usepackage[inline]{enumitem}

\definecolor{input}{HTML}{303060}
\definecolor{output}{HTML}{804000}
\definecolor{string}{HTML}{A02020}
\definecolor{parent}{HTML}{A020A0}
\definecolor{function}{HTML}{205080} 
\definecolor{constructor}{HTML}{205080}
\definecolor{method}{HTML}{205080}
\definecolor{keyword}{HTML}{008000}
\definecolor{error}{HTML}{B01010}
\definecolor{comment}{HTML}{60A060}

\newcommand{\defemph}[1]{\emph{#1}}

\newcommand{\noopsort}[1]{}

\DeclareMathOperator{\GL}{GL}
\DeclareMathOperator{\val}{val}

\newcommand{\ZZ}{\mathbb Z}

\newcommand{\R}{\mathbb R}

\newcommand{\RR}{\mathbb R}

\renewcommand{\mod}{\;\mathrm{mod}\;}

\newcommand{\ifnonempty}[3]{%
  \def\tempa{}%
  \def\tempb{#1}%
  \ifx\tempa\tempb 
  #3 
  \else            
  #2
  \fi}

\newcommand{\LT}{LT}

\newcommand{\LM}{LM}

\newcommand{\sage}{\textsc{SageMath}\xspace}

\definecolor{purple}{rgb}{0.6,0,0.6}

\definecolor{answer}{rgb}{0,0.5,0.2}

\newtheorem{theo}{Theorem}[section]
\newtheorem{lem}[theo]{Lemma}
\newtheorem{prop}[theo]{Proposition}

\theoremstyle{definition}
\newtheorem{rem}[theo]{Remark}
\newtheorem{ex}[theo]{Example}
\newtheorem{deftn}[theo]{Definition}

\copyrightyear{2020}
\acmYear{2020}
\setcopyright{acmlicensed}\acmConference[ISSAC '20]{International Symposium on Symbolic and Algebraic Computation}{July 20--23, 2020}{Kalamata, Greece}
\acmBooktitle{International Symposium on Symbolic and Algebraic Computation (ISSAC '20), July 20--23, 2020, Kalamata, Greece}
\acmPrice{15.00}
\acmDOI{10.1145/3373207.3404037}
\acmISBN{978-1-4503-7100-1/20/07}

\begin{document}

\title{On FGLM Algorithms with Tropical Gröbner bases}

\author{Yuki Ishihara}
  \affiliation{Graduate School of Science,
  \institution{Rikkyo University}
  \city{Tokyo, Japan}}
  \email{yishihara@rikkyo.ac.jp}

\author{Tristan Vaccon}
  \affiliation{Universit\'e de Limoges;
  \institution{CNRS, XLIM UMR 7252}
  \city{Limoges, France}  
  \postcode{87060}  
  }
  \email{tristan.vaccon@unilim.fr}

\author{Kazuhiro Yokoyama}
  \affiliation{Departement of Mathematics,
  \institution{Rikkyo University}
  \city{Tokyo, Japan}}
  \email{kazuhiro@rikkyo.ac.jp}

\begin{abstract}

Let $K$ be a field equipped with a valuation. Tropical varieties over $K$ can be defined with a theory of Gröbner bases taking into account the valuation of $K$.
Because of the use of the valuation, the theory of tropical Gröbner bases 
has proved to provide settings for 
computations over polynomial rings over a $p$-adic field
that are more stable than that of classical Gröbner bases.
In this article, we investigate how the FGLM change of ordering
algorithm can be adapted to the tropical setting.

As the valuations of the polynomial coefficients are taken into
account, the classical FGLM algorithm's incremental way, monomomial by monomial,
to compute the multiplication matrices and the change
of basis matrix can not be transposed at all to the tropical
setting. We mitigate this issue by developing
new linear algebra algorithms and apply them
to our new tropical FGLM algorithms.

Motivations are twofold.
Firstly, to compute tropical varieties, one usually goes
through the computation of many tropical Gröbner
bases defined for varying weights (and then varying term orders).
For an ideal of dimension $0$, the tropical FGLM algorithm provides an efficient way
to go from a tropical Gröbner basis from one weight to one for another weight.
Secondly, the FGLM strategy can be applied to go from a tropical
Gröbner basis to a classical Gröbner basis.
We provide tools to chain the stable computation of a tropical Gröbner
basis (for weight $[0,\dots,0]$) with the $p$-adic stabilized
variants of FGLM of \cite{RV:2016} to compute a lexicographical or shape position basis.

All our algorithms have been implemented into
\sage. We provide numerical examples to illustrate
time-complexity. We then illustrate
the superiority of our strategy regarding to the 
stability of $p$-adic numerical computations.
\end{abstract}



\vspace{-1.5mm}
\terms{Algorithms, Theory}


\setlength{\textfloatsep}{0pt}

\maketitle
\vspace{-1cm}
\section{Introduction}

The development of tropical geometry is now more than three decades old.
It has generated significant applications to very various domains, from algebraic geometry to combinatorics, computer science, economics, optimisation, non-archimedean geometry and many more. We refer to \cite{MS:2015} for
a complete introduction.

Effective computation of tropical varieties are now available
using Gfan and Singular (see \cite{JRS:2019} , \cite{GRZ:2019}).
Those computations often rely on the computation of so-called tropical Gröbner bases (we use \textit{GB} for Gröbner bases in the following).
Since Chan and Maclagan's definition of tropical Gröbner bases taking into account the valuation in \cite{CM:2019}, computations of 
tropical GB are available over fields with trivial or non-trivial valuation, using various methods: Matrix F5 in \cite{Vaccon:2015},
F5 in \cite{Vaccon:2017,Vaccon:2018} or lifting in \cite{MR:2019}.

An important motivation for studying the computation of
tropical GB is their numerical stability.
It has been proved in \cite{Vaccon:2015} that for polynomial 
ideals over a $p$-adic field, computing tropical 
GB (which by definition take into account the valuation), can be significantly
more stable than classical GB.

Unfortunately, no tropical term ordering can be an
elimination order, hence tropical GB can not be used directly for solving polynomial
systems.
Our work is then motivated by the following question: 
can we take advantage
of the numerical stability of the computation of tropical GB to compute a shape position
basis in dimension zero through a change of ordering algorithm?

In this article, we tackle this problem by studying
the main change of ordering algorithm, FGLM \cite{Faugere:1993}.
On the way, we investigate some adaptations and optimizations
of this algorithm designed to take advantage of some special
properties of the ideal (\textit{e.g.} Borel-fixedness of its initial ideal). 

We also provide a way to go from a tropical term order to another. 
This produces another motivation: difficulty of
  computation can vary significantly depending on
  the term order (see \S 8.1 of \cite{Vaccon:2018}),
  hence, 
  using a tropical FGLM algorithm, one could go
  from an easy term order to a harder one
  in an efficient way.

Finally, we conclude with numerical data to estimate the loss
 in precision for the computation of a lex Gröbner basis
  using a tropical F5 algorithm followed by an FGLM
   algorithm, in an affine setting, and also numerical data
    to illustrate the
behavior of the various variants of FGLM handled along the way.

\subsection{Related works}


Chan and Maclagan have developed in \cite{CM:2019} a Buchberger algorithm to compute tropical GB for homogeneous input polynomials (using a special division algorithm).
Following their work, adaptations of the F5 strategies
have been developped in  \cite{Vaccon:2015, Vaccon:2017, Vaccon:2018}
culminating with complete F5 algorithms for affine input polynomials.

A completely different approach has been developped by Markwig and Ren in \cite{MR:2019}, relating the computation of tropical GB
in $K[X_1,\dots, X_n]$ to the computation of standard basis in
$R\llbracket t \rrbracket [X_1,\dots, X_n]$ (for $R$ a subring of the ring of
integers of $K$).
It can be connected to the Gfanlib interface in Singular to compute 
tropical varieties (see: \cite{JRS:2019}).

Finally, Görlach, Ren and Zhang have developped in \cite{GRZ:2019} a way to compute
zero-dimensional tropical varieties using shape position bases
and projections. Their algorithms take as input
a lex Gröbner basis in shape position.
Our strategies can be used to provide such a basis stably (precision-wise)
when working with $p$-adic numbers, 
and be chained with their algorithms.

\subsection{Notations}

Let $K$ be a field with a discrete valuation $\val$ such that
$K$ is complete with respect to the norm defined by $\val$. We denote by $R=O_K$
its ring of integers, $m_K$ its maximal ideal (with $\pi$ a uniformizer), and $k=O_K/m_K$ its fraction
field. We
refer to Serre's Local Fields \cite{Serre:1979} for an introduction to such
fields. 
Classical
examples of such fields are $K = \mathbb{Q}_p$, with $p$-adic valuation, and
$\mathbb{Q}((X))$ or $\mathbb{F}_q((X))$ with $X$-adic valuation. 


The polynomial ring $K[X_1,\dots, X_n]$ (for some $n \in \mathbb{Z}_{> 0}$) will be denoted by $A$, and
for $u=$ $(u_1,\dots$ $,u_n)$ $\in \mathbb{Z}_{\geq 0}^n$, we write $x^u$ for
$X_1^{u_1} \dots X_n^{u_n}.$
For $g \in A,$ $\vert g \vert$ denotes the total degree of $g$
and $A_{\le d}$ the set of all polynomials in $A$ of total degree less than $d$. 
The matrix of a finite list of polynomials (of total degree $\leq d$ for some $d$) written in a basis of monomials (of total degree $\leq d$)
is called a \defemph{Macaulay matrix}.

For $w \in Im(\val)^n \subset \mathbb{R}^n$ and $\leq_m$ a monomial order on $A,$
we define $\leq$ a tropical term order as in the following definition:

\begin{deftn} \label{defn:trop_term_order}
Given $a,b \in K^*=K \setminus \{ 0 \}$ and $x^\alpha$ and $x^\beta$ two monomials in $A$, 
we write $a x^\alpha < b x^\beta$ if:
\begin{itemize}
\item $\vert x^\alpha \vert < \vert x^\beta \vert,$ or
\item $\vert x^\alpha \vert = \vert x^\beta \vert,$
and
$\val(a)+w \cdot \alpha > \val(b) +w \cdot \beta$, or
\item $\vert x^\alpha \vert = \vert x^\beta \vert,$
$\val(a)+w \cdot \alpha = \val(b) +w \cdot \beta$ 
and $x^\alpha <_m x^\beta.$
\end{itemize} 
For $u$ of valuation $0,$ we write $a x^\alpha =_{\leq} u a x^\alpha.$
Accordingly, $a x^\alpha \leq b x^\beta$
if $a x^\alpha < b x^\beta$ or $a x^\alpha =_{\leq} b x^\beta.$
\end{deftn}

Leading terms ($\LT$) and leading monomials ($\LM$) are defined according to this
term order. See Subsec. 2.3 of \cite{Vaccon:2018} for more information
on this definition and its comparison with Def. 2.3 of \cite{CM:2019}.

Let $I \subset A$ be a $0$-\textit{dimensional}.
Let $B_{\leq }$ the canonical linear $K$-basis of $A/I$ made of
the $x^\alpha \notin \LM_{\leq}(I)$. Let $\delta$ be the cardinality of $B_{\leq }.$ We denote by $\mathscr{B}_{\leq}$ the border
of $B_{\leq }$ (\textit{i.e.} the $x_k x^\alpha$ for $k \in \llbracket 1,n \rrbracket$ such that 
$x^\alpha \in B_{\leq }$ and $x_k x^\alpha$ not in $B_{\leq }$).
$NF_{\leq}$ is the normal form mapping defined by $I$ and $\leq.$ 
We define $D$ such that $D=1+\max_{x^\alpha \in B_\leq} \vert x^\alpha \vert $.

\section{Multiplication matrices} \label{section:mult_mat}

The first task in the FGLM strategy is to develop
the tools for computations in $A/I.$
The main ingredients are the multiplication matrices,
$M_1,\dots, M_n$, corresponding to 
the matrices of the linear maps given by the multiplication
by $x_i$ written in the basis $B_{\leq }$.

Once they are known, it is clear that one can perform
any $K$-algebra operation on elements of $A/I$
written in the basis $B_{\leq }$.

To compute those matrices, a natural strategy is to
go through the computation of the normal forms
$NF(x_i x^\alpha)$ for $x^\alpha \in B_\leq.$

We investigate in this section how to
proceed with this task, and how
it compares to the classical case.

\subsection{Linear algebra}

We recall here the tropical row-echelon form algorithm
of \cite{Vaccon:2015}
that we use for computing normal forms using linear algebra.

\begin{algorithm} 

 \SetKwInOut{Input}{input}\SetKwInOut{Output}{output}

 \Input{$M$, a Macaulay matrix of degree $d$ in $A$, with $n_{row}$ rows and $n_{col}$ columns, and $mon$ a list of monomials indexing the columns of $M.$}
 \Output{$\widetilde{M}$, the $U$ of the tropical LUP-form of $M$}

$\widetilde{M} \leftarrow M$   \;
		
\For{$i=1$ to $n_{row}$}{
\textbf{Find} $j$ such that $\widetilde{M}[i,j]$ has the greatest term $\widetilde{M}[i,j] x^{mon_j}$ for $\leq$ of the row $i$ \;
\textbf{Swap} the columns $i$ and $j$ of $\widetilde{M}$, and the $i$ and $j$ entries of $mon$ \;
By \textbf{pivoting} with the $i$-th row, eliminates the coefficients of the other rows on the first column; \;
}
\textbf{Return} $\widetilde{M}$ \;

 \caption{The tropical row-echelon form algorithm} \label{algo:trop_LUP}
\end{algorithm}

We refer the interested reader to \cite{Vaccon:2015, Vaccon:2018}.
We illustrate this algorithm with the following example.

\begin{ex}
We present the following Macaulay matrices, over
$\mathbb{Q}_3[x,y]$ with $w=(0,0),$
and $\leq_m$ be the graded lexicographical
ordering.
The second one is the output of the tropical LUP algorithm applied on the first one.
The monomials indexing the columns are written on top of the matrix.

\begin{tikzpicture}
\matrix (m) [matrix of math nodes,nodes in empty cells, ampersand replacement=\& ] at (0,0) {  x^4 \& x^3 y \& y^4 \& x^2  \& xy \& y^2 \\ 
 1  \& \& \& \& 3 \& \phantom{2} \\
  \& \& \& 1 \& 9 \& 3 \\
   \& 9 \& 9\&\&\&   \\
  \phantom{1}  \& 9 \& 9 \& 3 \& 1 \& 9 \\ 
} ;
\draw (m-2-1.north west) to[ right] (m-5-1.south west);
\draw (m-2-6.north east) to[ left] (m-5-6.south east);

\matrix (n) [matrix of math nodes,nodes in empty cells, ampersand replacement=\& ] at (4,0) {  x^4 \& x^2 \& x^3 y \& xy  \& y^4 \& y^2 \\ 
1  \& \& \& \& 3 \& \phantom{-18} \\
  \& 1 \& \& 0 \&  \&-\frac{57}{35} \\
   \&  \& 9\& 0 \& 9 \& -\frac{162}{35}  \\
   \phantom{1}  \&  \&  \& -35 \& 0 \& -18 \\ 
} ;
\draw (n-2-1.north west) to[ right] (n-5-1.south west);
\draw (n-2-6.north east) to[ left] (n-5-6.south east);
\end{tikzpicture}. 

If all four polynomials represented by the matrix belong to some ideal $I$
(and assuming that $y^4, y^2 \in B_\leq (I)$)
then we can conclude that $NF_\leq ( xy)=- \frac{18}{35} y^2$
and $NF_\leq ( x^3y)=-y^4+\frac{18}{35}y^2.$
\end{ex}

\subsection{Comparison with classical case}

The classical strategy to compute the $NF_{\leq_m}(x_i x^\alpha)$ ($x^\alpha \in B_{\leq_m}$) when working with a monomial ordering $\leq_m$, starting with a reduced GB $G$, is to set apart the following only three cases possible:
\begin{description*} \item[(Type 1)] $x_i x^\alpha \in B_{\leq_m}$; \item[(Type 2)] $x_i x^\alpha \in \LT (G)$; \\
\item[(Type 3)] $x_i x^\alpha \in \LT_{\leq_m} (I)$ but neither in $ B_{\leq_m}$ nor in $\LT (G)$.  \end{description*}

Type 1 is the easiest, as in this case $ NF_{\leq_m}(x_i x^\alpha) = x_i x^\alpha$.
Type 2 is not very difficult either.
If for some $g \in G,$ $\LM(g) =  x_i x^\alpha$, $g=x_i x^\alpha+\sum_{x^\beta \in B_{\leq_m}} c_\beta x^\beta$,  then as $G$ is reduced, we get directly that
$NF_{\leq_m}(x_i x^\alpha)=-\sum_{x^\beta \in B_{\leq_m}} c_\beta x^\beta$.

Type 3 is the trickiest. We assume that we have already computed all the $NF(x_j x^\beta)$
for $x_j x^\beta <_m x_i x^\alpha.$
Let $x_k$ be the smallest (for $\leq_m$) variable dividing $x_i x^\alpha$.
Then the normal form \[NF \left(\frac{x_i x^\alpha}{x_k} \right)=\sum_{x^\beta \in B_{\leq_m}, \: x^\beta <_m \frac{x_i x^\alpha}{x_k}} c_\beta x^\beta\] is already known.
As in the previous sum, $x^\beta <_m \frac{x_i x^\alpha}{x_k}$, then $x_k x^\beta <_m x_i x^\alpha$, and all the $NF(x_k x^\beta)$'s
are also already known. 
Therefore, we can write
\[NF(x_i x^\alpha)=\sum_{x^\beta \in B_{\leq_m}, \: x^\beta <_m \frac{x_i x^\alpha}{x_k}} c_\beta NF(x_k x^\beta),\] and $NF(x_i x^\alpha)$ can be obtained from the previous normal forms.

It is easy to see that the cost of computation of a normal form in the third case
is in $O(\delta^2)$ field operations. The other two cases are negligible.
As there are $O(n \delta)$ multiples to consider,
the total cost for the computation of the multiplication matrices is in
$O(n\delta^3)$ field operations.

Unfortunately, this strategy can not be completely generalized to the tropical context.
There is no issue with the first two computations.
However, there is no straightforward way to adapt the third one. 
We illustrate this failure with the following example.

\begin{ex}
Over $\mathbb{Q}_3[x,y]$ with $\leq$ defined by $w=(0,0),$
and $\leq_m$, the graded lexicographical
ordering, let us take $I=\left\langle f_1,f_2,f_3,f_4 \right\rangle$
with $f_1=x^7,$ $f_2=x^4y^2+3x^5y+12x^3y^3+9xy^5,$ $f_3=x^2 y^4+9x^5 y+18 x^3 y^3+9 x y^5,$
$f_4=y^6+12x^5y+3x^3y^3+6xy^5.$
The first monomials of the third type arrive in degree $7,$
namely $xy^6, x^2 y^5, x^4 y^3, x^5 y^2$.
Due to the fact that we use a tropical term order, 
$f_2,$ $f_3,$ and $f_4$ all involve the monomials $x^5 y, x^3 y^3, x y^5.$
In consequence if one wants to use multiples of the $NF(x^4y^2),$ $NF(x^2 y^4),$ $NF(y^6),$
one gets quantity involving each three monomials among  $xy^6, x^2 y^5, x^4 y^3,$ and $x^5 y^2$.
They are all intertwined, and the trick we saw previously for
monomials of the third type can not be used.
\end{ex}

\subsection{Tropical GB: General case}

To untangle the reduction of monomials of the third type, we can use linear algebra.
We have to proceed degree by degree.
While monomials of the first type do not need any special proceeding,
we need to interreduce the reductions of the monomials of the second
and third types.
The general strategy is described in Algorithm \ref{algo:mult_mat}.

\begin{algorithm} 
  \SetKwInOut{Input}{input}\SetKwInOut{Output}{output}
  \Input{A reduced GB $G$ of the ideal $I$ for $\leq$, a tropical term ordering.}
  \Output{$M_1,\dots,M_n$ the multiplication matrices of $A/I$ (over the basis $B_\leq$).}

  Using $LT(G)$, compute $B_\leq$ (and $\delta = \sharp(B_\leq)$)\;
  Define $M_1,\dots,M_n$ as zero matrices in $K^{\delta \times \delta}$, their rows and columns
  are indexed by the $x^\alpha \in B_\leq$  \;
  Compute $L=\{x_i x^\alpha, \textrm{ for } i \in \llbracket 1,n \rrbracket \textrm{ and } x^\alpha \in B_\leq \}.$ \;
  Compute $\overline{L}=L \cap (B_\leq \cup LT(G))^c$ \;
  \For{$x^\alpha \in L \cap B_\leq$}
  		{\For{$i$ such that $x_i$ divides $x^\alpha$}
  				{Set $M_i[x^\alpha, \frac{x^\alpha}{x_i}]=1$ \;
  				\tcc{The column indexed by $\frac{x^\alpha}{x_i}$ is zero, except on its coefficient indexed by $x^\alpha/x_i$ }}	  
  		}
  \For{$x^\alpha \in L \cap LT(G)$}
  		{	Take $g \in G$ such that $g$ can be written $g=x^\alpha +\sum_{x^\beta \in B_\leq} g_{x^\beta} x^\beta$ \;
  			\For{$i$ such that $x_i$ divides $x^\alpha$}
  				{
				 \For{$x^\beta \in B_\leq$}
				 		{Set $M_i[x^\beta, \frac{x^\alpha}{x_i}]=-g_{x^\beta}$ \; }  				
  				}	  
  		}
  		
  Set $\mathscr{M}$ to be a matrix over $K$ with $0$ rows and with columns indexed by $\overline{L} \cup LT(G) \cup B_\leq.$ \;
  
  \For{$d$ a degree of a monomial in $\overline{L}$ (in ascending order)}
  {
  \For{$x^\alpha \in \overline{L}$ of degree $d$}
	  {
		Find $x_i$, and $g$ either in $G$ or as a row of $\mathscr{M}$ such that $LT(x_i g)=x^\alpha$ \;
		Stack $x_i g$ at the bottom of $\mathscr{M}$ \; 	  
	  }
  Using multiples of the form $x_i g$ or $g$, for $g$ either in $G$ or as a row of $\mathscr{M}$, find a complete set of reducers for all the monomials in $\overline{L} \cup LT(G)$ appearing with a non-zero coefficient in their column, and stack them  at the bottom of $\mathscr{M}$ \;
  
  Compute the Tropical Row-echelon form of $\mathscr{M}$ by Algorithm \ref{algo:trop_LUP} and replace $\mathscr{M}$ with it \;
  
  \For{$x^\alpha \in \overline{L}$}
  		{	Take the row $s$ of $\mathscr{M}$ with leading coefficient $x^\alpha.$ \;
  			\For{$i$ such that $x_i$ divides $x^\alpha$}
  				{
				 \For{$x^\beta \in B_\leq$}
				 		{Set $M_i[x^\beta, \frac{x^\alpha}{x_i}]=-\frac{\mathscr{M}[s,x^\beta]}{\mathscr{M}[s,x^\alpha]}$ \; }  				
  				}	  
  		}  
  }
  \textbf{Return $M_1,\dots,M_n$}
  
  \caption{Multiplication matrices computation algorithm} \label{algo:mult_mat}
\end{algorithm}

\begin{prop}
Algorithm \ref{algo:mult_mat} is correct, and is in $O(n^3 \delta^3)$ field operations over $K$.
\end{prop}
\begin{proof}
The essentially different part compared to the classical
case starts on Line 13. Lines 16 and 18 are crucial.
By definition, monomials of the third type are in $\LT(I).$
If $x^\alpha \in \overline{L}$ can not be written as $x_k x^\beta$ with 
$x^\beta$ of type 2 or 3, it means that all its divisors are in $B_{\leq}$.
Consequently, it is a minimal generator ot $\LT(I)$ and is of type 2, which is a contradiction.
Therefore, any monomial of the third type is a simple multiple of a monomial
of type 2 or 3.

As in the \textbf{for loop} on Line 14, we proceed by increasing degree,
it is an easy induction to prove that such desired $x_i$ and $g$ exist.

For the complete set of reducers on Line 18, we use the fact that the monomials 
appearing in $\mathscr{M}$ all are in $B_\leq \cup L,$ again by an easy induction (using the fact
that the rows of $\mathscr{M}$ in previous degree are already reduced), and
therefore, the complete set of reducers can be built.

The Tropical Row-echelon form computation then produces the desired 
normal forms. The correctness is then clear.

Regarding to the arithmetic complexity, we should
note that both rows and columns of $\mathscr{M}$ are indexed by monomials in
$L \cup B_\leq$ and there are $O(n \delta)$ of them.
With the row-reduction, the total cost is then in $O(n^3 \delta^3)$ arithmetic operations.
\end{proof}

\begin{rem}
The matrix $\mathscr{M}$ is sparse: any row added to the matrix on Line 17 has at most $\delta+1$ non-zero coefficients: it is obtained as the multiple of a reduced row. Can we take advantage of this $\frac{1}{n}$ sparsity ratio for a better complexity?
\end{rem}

\begin{ex}
Let $G=(y+2x,x^2+4)$ be a GB for $w=[0,0]$ and grevlex of the ideal it spans in
$\mathbb{Q}_2[x,y].$ Then $B_\leq = \{1,x \},$ $L=\{x,y,x^2,xy \}$ and
$\overline{L}=\{xy \}.$ Only $d=2$ is considered on Line 4 of Algorithm \ref{algo:mult_mat}.
The following matrices represent respectively $\mathscr{M}$ before and after applying Algorithm \ref{algo:trop_LUP}, $M_1$ and $M_2$:

\begin{tikzpicture}
\matrix (m) [matrix of math nodes,nodes in empty cells, ampersand replacement=\& ] at (0,0) {  x^2 \& x y \& 1 \\ 
 2  \& 1\&0  \\
 1 \&0 \& 4 \\
} ;
\draw (m-2-1.north west) to[ right] (m-3-1.south west);
\draw (m-2-3.north east) to[ left] (m-3-3.south east);

\matrix (n) [matrix of math nodes,nodes in empty cells, ampersand replacement=\& ] at (2,0) {  x y  \& x^2 \& 1 \\ 
 1  \&0 \& -8  \\
0  \&1 \& \phantom{-} 4 \\
} ;
\draw (n-2-1.north west) to[ right] (n-3-1.south west);
\draw (n-2-3.north east) to[ left] (n-3-3.south east);

\matrix (o) [matrix of math nodes,nodes in empty cells, ampersand replacement=\& ] at (4,0) {  (x*) \& 1 \& x \\ 
1   \&0 \&-4 \\
 x \&1 \& \phantom{-}0 \\
} ;
\draw (o-2-2.north west) to[ right] (o-3-2.south west);
\draw (o-2-3.north east) to[ left] (o-3-3.south east);

\matrix (p) [matrix of math nodes,nodes in empty cells, ampersand replacement=\& ] at (6,0) {  (y*) \& 1 \& x \\ 
1   \& \phantom{-}0 \&8 \\
 x \&-2  \&  0 \\
} ;
\draw (p-2-2.north west) to[ right] (p-3-2.south west);
\draw (p-2-3.north east) to[ left] (p-3-3.south east);

\draw (.8,-.5) node{,} ;
\draw (2.9,-.5) node{,};
\draw (4.95,-.5) node{,};
\draw (7,-.5) node{.};
\end{tikzpicture} 
\end{ex}

\subsection{Finite precision}

We can now analyze the loss in precision when
applying Algorithms \ref{algo:trop_LUP} and \ref{algo:mult_mat}.
To prevent loss in precision to explode exponentially,
we replace Line 5 of Algorithm \ref{algo:trop_LUP} with the following
two rows:
\begin{enumerate}
\item By pivoting using the 'leading terms' of the rows $j$ for $j>i$, eliminate all the coefficients possible of row $i$ ;
\item By pivoting with row $i$, eliminate all the coefficients on the $i$-th column. 
\end{enumerate}

The first row makes sense because by construction, all the rows of $\mathscr{M}$
have distinct leading terms, and this is kept unchanged during the pivoting process.

\begin{prop}
Let us assume that the matrix built on Line 17 of Algorithm \ref{algo:mult_mat}
has coefficients in $K$ known at precision $O(\pi^N)$.
All rows have distinct leading terms, leading coefficient $1$
and let us take $\Xi$ be the smallest valuation of a coefficient
of this matrix $\mathscr{M}$. We assume that $\Xi \leq 0$.
Let $l = rank(\mathscr{M}).$
We assume that $N>-l^2 \Xi.$
Then, after the application of Algorithm \ref{algo:trop_LUP}\footnote{using the modification
presented just above this proposition}, the coefficients
of the obtained matrix $\tilde{\mathscr{M}}$ are known at precision $O(\pi^{N+l^2 \Xi})$,
and the smallest valuation of a coefficient $\tilde{\mathscr{M}}$ is lower-bounded by
$l \Xi.$
\end{prop}
\begin{proof}
After the reduction of row 1 by the other rows, the smallest valuation on
row 1 is lower-bounded by $l \Xi$ and its coefficients are known at precision at least $O(\pi^{N+l \Xi}).$
The coefficients of row 1 for the columns indexed by $\overline{L}\cup \LT (G)$ are all zeros, except
for its leading coefficient, which is $1+O(\pi^{N+(l-1) \Xi}).$
After the reduction of the other rows by row 1, on the rows of index $>1$,
the coefficients for the columns indexed by $\overline{L}\cup \LT (G)$
are of valuation at least $\Xi$ and known at precision $O(\pi^{N+l \Xi}).$
The coefficients for the columns indexed by $B_\leq$
are of valuation at least $l\Xi$ and known at the same precision.
The desired result follows by an easy induction argument. 
\end{proof}

We then upper-bound the loss in precision
for the whole computation of the multiplication matrices.
Recall that: $D=1+\underset{x^\alpha \in B_\leq}{\max} \vert x^\alpha \vert $.

\begin{prop}
Let us assume that the smallest valuation of a coefficient of $G$
is $\Xi$ and that the coefficients of $G$ are known at precision $O(\pi^N)$.
As $G$ is reduced, we get that $\Xi \leq 0$.

Then the coefficients of the matrices $M_1,\dots,M_n$ are of
valuation at least $\left( n \delta \right)^D \Xi$,
and are known at precision $O \left(\pi^{N+\left( \frac{(n \delta)^{2D+2}-1}{(n \delta)^2-1} \right) \Xi} \right).$
\end{prop}
\begin{proof}
This is a corollary to the previous proposition.
There are at most $D$ calls to the previous proposition, with matrices of
ranks $l_1,\dots,l_D$.
Consequently, the upper bound on the valuation is
$l_1  \dots  l_D  \Xi$
and the precision is in $O(\pi^{N+(l_1^2+l_1^2 l_2^2+\dots+l_1^2\dots l_D^2) \Xi})$
which is in $O(\pi^{N+D(l_1^2\dots l_D^2) \Xi})$
As for all $i,$ $l_i \leq n\delta$,
we get the desired bounds.
\end{proof}

\begin{rem}
In the very favorable case where $G$ is homogeneous and 
$w=[0, \dots,0]$, we get that $\Xi =0$, and no loss in precision
is happening.
This is unfortunately not the most interesting case for polynomial system
solving.
Numerical data in Section \ref{sec:num_data} will show
that loss in precision remain very reasonnable when using $w=[0, \dots,0]$ even in the affine case.
\end{rem}


\subsection{Using semi-stability} 
\label{subsec:using_semi_stability}

Following Huot's PhD thesis  \cite{Huot:13}, when Borel-fixedness (see Subsec. \ref{subsec:borel_fixedness}) or semi-stability properties are
satisfied, many arithmetic operations can be avoided during the computation of the multiplication matrices.
We begin with semi-stability.

\begin{deftn}
$I$ is said to be semi-stable for $x_n$ if for all $x^\alpha$ such that $x^\alpha \in LM(I)$ and $x_n \mid x^\alpha$ we have for all $k \in \llbracket 1, n-1 \rrbracket$ $\frac{x_k}{x_n} x^\alpha \in LM(I).$
\end{deftn}

Semi-stability's application is explained in Proposition 4.15, Theorem 4.16 and Corollary 4.19 of \cite{Huot:13} (see also Section 4 of \cite{Faugere:2014}). We recall the main idea here with its adaptation to the
tropical setting:
\begin{prop} \label{prop:Huot} 
If $I$ is semi-stable for $x_n,$ $M_n$ can be read from $G$ and requires no arithmetic operation.\end{prop}
\begin{proof}
The proof is the same as that of Theorem 8 of \cite{Faugere:2014}.
We prove that $\overline{L}\cap x_n B_\leq=\emptyset.$
Let  $x_n x^\alpha \in \overline{L}\cap x_n B_\leq,$ with $x^\alpha \in B_\leq.$
Then there is some monomial $m$ and $g \in G$ such that $LM(mg)=x_n x^\alpha .$
As $x^\alpha \in B_\leq,$ we get that $x_n \nmid m$.
Since $x_n x^\alpha \in \overline{L},$ then $\vert m \vert \geq 1.$
Let $k<n$ be such that $x_k \mid m.$
Then, by semi-stability for $x_n$,
$x^\alpha = \frac{m}{x_k} \times \frac{x_k LM(g)}{x_n} \in LM(I),$
which is a contradiction.
\end{proof}

\begin{algorithm} 
  \SetKwInOut{Input}{input}\SetKwInOut{Output}{output}
  \Input{A reduced GB $G$ of the ideal $I$ for $\leq$, a tropical term ordering, assuming $I$ is semi-stable for $x_n$}
  \Output{$M_n$ the matrix of the multiplication by $x_n$ in $A/I$}

  Using $LT(G)$, computes $B_\leq$ (and $\delta = \sharp(B_\leq)$)\;
  Define $M_n$ as a zero matrix in $K^{\delta \times \delta}$, its rows and columns
  are indexed by the $x^\alpha \in B_\leq$  \;
  Compute $L_n=\{x_n x^\alpha, \textrm{ for } x^\alpha \in B_\leq \}.$ \;
  \For{$x^\alpha \in L_n \cap B_\leq$}
  		{ Set $M_n[x^\alpha, \frac{x^\alpha}{x_n}]=1$ \;}	  
  		
  \For{$x^\alpha \in L_n \cap LT(G)$}
  		{	Take $g \in G$ such that $g$ can be written $g=x^\alpha +\sum_{x^\beta \in B_\leq} g_{x^\beta} x^\beta$.
  			\For{$x^\beta \in B_\leq$}
				 {Set $M_n[x^\beta, \frac{x^\alpha}{x_i}]=-g_{x^\beta}$ \; }  				
  		}	  
  \textbf{Return} $M_n$ \;
  
  \caption{Computing $M_n$, when semi-stable for $x_n$} \label{algo:M_n_when_semi_stable}
\end{algorithm}

Thanks to Proposition \ref{prop:Huot}, Algorithm \ref{algo:M_n_when_semi_stable} is correct,
and its arithmetic cost is given by the following proposition.

\begin{prop}
Given a reduced GB $G$ of the ideal $I$ for $\leq$, a tropical term ordering, and assuming $I$ is \textit{semi-stable for} $x_n$, then $M_n$ can be computed in $O(\delta^2)$ arithmetic operations, which are only computing opposites.
\end{prop}

To apply the previous result to compute a GB in shape position in Subsection \ref{subsec:trop_to_shape}, we need
to also compute the $NF(x_i)$'s. The following lemma states that this is not costly.

\begin{lem}
Given a reduced GB $G$ of the ideal $I$ for $\leq$, a tropical term ordering, then the $NF_\leq (x_i)$'s
 can be computed in $O(n\delta)$ arithmetic operations, which are only computing opposites.
\end{lem}
\begin{proof}
It is a consequence of the fact that $\leq$ is degree-compatible: for any $i$, $x_i$ is
either in $\LT(G)$ or in $B_\leq$.
\end{proof}

Subsection \ref{subsec:trop_to_shape} will apply the previous two
results to obtain a fast algorithm to compute a shape-position basis. 

\begin{rem}
For grevlex in the classical case, it is known that after a generic change of variable, $I$ is semi-stable for $x_n.$
The reason is that after a generic change of variable, $LT(I)$ is equal to the GIN of $I$ (see Definition 4.1.3 of \cite{HH:2011}) , which
is known to be Borel-fixed, and Borel-fixedness implies semi-stability for $x_n.$
In Section \ref{sec:gin}, we investigate whether this strategy is still valid
in the tropical case.
\end{rem}

\section{GIN and Borel-fixed initial ideal}
\label{sec:gin}

In this section, we introduce the tropical generic initial ideal of a $0$-dimensional ideal analogously to the classical case, and study its properties of Borel-fixedness and semi-stability. 
The desired goal is to be able to use the fast Algorithm \ref{algo:M_n_when_semi_stable} after a (generic) change of variable.

\subsection{Tropical GIN}

We follow the lines of Chapter 4 of \cite{HH:2011}, and use the usual action of $\GL_n(K)$ on
$A$: $(\eta,f(x)) \in \GL_n(K) \times A \mapsto \eta(f):=f(\eta^\top \cdot x).$

\begin{deftn}
    An external product of monomials $x^{\alpha_1}\wedge \cdots \wedge x^{\alpha_k}$ is called a {\em standard exterior monomial} if $x^{\alpha_1}\geq \cdots \geq x^{a_k}$. If its monomial is standard, a term $c  x^{\alpha_1}\wedge \cdots \wedge x^{\alpha_k}$ is called a {\em standard exterior term}.
    We define an ordering on standard exterior terms by setting that:
    $cx^{\alpha_1} \wedge \cdots \wedge x^{\alpha_k}\geq dx^{\beta_1}\wedge \cdots \wedge x^{\beta_k}$ if $\val (c)+\sum_{i=1}^k w \cdot \alpha_i<\val (d)+\sum_{i=1}^k w\cdot \beta_i$, or $\val (c)+\sum_{i=1}^k w\cdot \alpha_i=\val (d)+\sum_{i=1}^k w\cdot \beta_i$ and there exists $1\leq j\leq k$ s.t. $x^{\alpha_j} > x^{\beta_j}$ and $x^{\alpha_i}=x^{\beta_i}$ for all $i<j$.
    We then define the leading term of an external product 
    of polynomials $f_1 \wedge \dots \wedge f_k$
    as its largest term, and denote it by $\LT ( f_1 \wedge \dots \wedge f_k).$
    The monomial of the leading term is denoted by $\LM (f_1 \wedge \dots \wedge f_k).$
\end{deftn}

\begin{lem} \label{lem:initial}
  Let $(f_1,\ldots,f_t) \in A^t$. If $\LT (f_1)> \cdots > \LT (f_t)$, then $\LT (f_1\wedge \cdots \wedge f_t)=\LT (f_1)\wedge \cdots \wedge \LT (f_t)$. 
\end{lem}
\begin{proof}
    Let $c_i$ be the coefficient of $\LM (f_i)$ in $f_i$. Then, $c=\prod c_i$ is the coefficient of $\Gamma=\LT (f_1)\wedge \cdots \wedge \LT (f_t)$ in $f_1\wedge \cdots \wedge f_t$.  We may assume that the $f_i$'s are ordered such that
$c\LT (f_1)\wedge \cdots \wedge \LT (f_t)$ is a standard exterior term.
Let $\Delta=dv_1\wedge \cdots \wedge v_t$ be another term in $f_1\wedge \cdots \wedge f_t$ and $d_i$ the coefficient of $v_i$ in  $f_i$.  Let $x^{\alpha_i}=\LM (f_i)$ and $x^{\beta_i}=v_i$. Since $c_i x^{\alpha_i}$ is the leading term of $f_i$,  it follows that $\val (c_i)+ w\cdot \alpha_i \le \val (d_i)+w\cdot \beta_i$. Thus, $\sum_{i=1}^t (\val (c_i)+ w\cdot \alpha_i)\le \sum_{i=1}^t (\val (d_i)+w\cdot \beta_i)$. As $\val(c)=\sum_{i=1}^t c_i$ and $\val(d)=\sum_{i=1}^t d_i$, we obtain $\val (c)+\sum_{i=1}^k w\cdot \alpha_i\le \val (d)+\sum_{i=1}^k w\cdot \beta_i$. If the inequality is strict
then $\Gamma$ is strictly bigger than any permutation of the monomials
of $\Delta$ such that a standard exterior term is obtained.
If equality holds. Then, for all $i$, $\val (c_i)+ w\cdot \alpha_i = \val (d_i)+w\cdot \beta_i$  and $x^{\alpha_i} \geq x^{\beta_i}$.
As $\Gamma$ is a standard exterior term, we deduce that also in this case,
$\Gamma$ is strictly bigger than any permutation of the monomials
of $\Delta$ such that a standard exterior term is obtained. 
\end{proof}

\begin{lem} \label{lem:basis}
	Let $V\subset A$ be a $t$-dimensional $K$-vector space.
    Let $w_1,\ldots,w_t$ be monomials with $w_1 >\cdots > w_t$. Then the following conditions are equivalent. 
    \begin{enumerate}
        \item the monomials $w_1,\ldots,w_t$ form a $K$-basis of $\LT (V)$,
        \item if $(f_1,\ldots,f_t)$ is a $K$-basis of $V$, then $\LM (f_1\wedge \cdots \wedge f_t)=w_1\wedge \cdots \wedge w_t$, 
        \item there exists a $K$-basis $(f_1,\ldots,f_t)$ of $V$ s.t. $\LM (f_1\wedge \cdots \wedge f_t)=w_1\wedge \cdots \wedge w_t$. 
    \end{enumerate}
\end{lem}

\begin{proof}
    $(1)\Rightarrow (2)$: We may assume that the $f_j$'s are monic and  $\LT (f_1)> \cdots > \LT (f_t)$. Since $\LT (f_i)\in \LT (V)$, there is $j(i)$ s.t. $\LT (f_i)=w_{j(i)}$. As $w_1>_1\cdots >_1 w_t$, we obtain $j(i)=i$ and $\LT (f_i)=w_i$ for all $i$. By Lemma \ref{lem:initial}, $\LT (f_1\wedge \cdots \wedge f_t)=\LT (f_1)\wedge \cdots \wedge \LT (f_t)=w_1\wedge \cdots \wedge w_t$. 
    
    $(2)\Rightarrow (3)$: It is obvious by choosing a $K$-basis $f_1,\ldots,f_t$ of $V$. 
    
    $(3)\Rightarrow (1)$: Since $\dim(V)=\dim(\LT ( V))$ and $w_1,\ldots,w_t$ is linear independent, it is enough to show that $w_i\in \LT (V)$. Let $f_1,\ldots,f_t$ be monic polynomials forming a $K$-basis of $V$ with $\LT (f_1)> \cdots > \LT (f_t)$ and $\LT (f_1\wedge \cdots \wedge f_t)=w_1\wedge \cdots \wedge w_t$. By Lemma \ref{lem:initial}, $\LT (f_1\wedge \cdots \wedge f_t)=\LT (f_1)\wedge \cdots \wedge \LT (f_t)$ and thus $w_i\in \LT (V)$. 
\end{proof}


\begin{prop} \label{prop:3-1}
Let $V\subset A_d$ be a $t$-dimensional $K$-vector space and $f_1,\ldots,f_t$ a basis of $V$. 
    Let $c w_1\wedge \cdots \wedge w_t$ be the largest (up to multiplication by an element of valuation $0$) standard exterior term of $\bigwedge^t A_{\le d}$  such that there exists $\eta \in \GL_n (R)$ with
    \[
    \LT (\eta (f_1)\wedge \cdots \wedge  \eta (f_t))=c w_1\wedge \cdots \wedge w_t. 
    \]
     Let $U_V=\{\eta  \in \GL_n(R)\mid \LT (\eta (f_1)\wedge \cdots \wedge  \eta (f_t))= \varepsilon \times c w_1\wedge \cdots \wedge w_t, \: \val (\varepsilon)=0\}$. Then, $U_V$ is open
     in $\GL_n (R)$ and  for any $\eta, \upsilon \in U_V$, $\LT (\eta V)=\LT (\upsilon V)$.        
\end{prop}
\begin{proof}
As only a finite amount of monomials are possible and $\val (R)$ is discrete and $\geq 0$, $U_V$ is well-defined.
The valuation being discrete, $U_V$ is open:
$\LT (\eta (f_1)\wedge \cdots \wedge  \eta (f_t))=\varepsilon \times c w_1\wedge \cdots \wedge w_t$ amounts to $\val (q(\eta)) < \nu$ for carefully chosen
$\nu \in \RR$ and polynomial $q \in \mathbb{Z}[k^{n \times n}].$
The last statement follows from Lemma \ref{lem:basis}.
\end{proof}

From Lemma \ref{lem:basis}, $w_1\wedge   \cdots \wedge w_t$ in Prop \ref{prop:3-1} is independent of the choice of basis of $V$. 
For $d \in \ZZ_{\geq 0},$ let $I_{\le d}=I\cap A_{\le d}$.
\begin{theo} \label{theo:tropical gin}
     Let $I$ be a $0$-dimensional ideal with $\delta =\dim_K K[X]/I$. We consider the finite dimensional $K$-vector space $I_{\le \delta}$. 
     Then the non-empty open set $U_I:=U_{I_{\le \delta}} \subset  \GL_n (R)$ satisfies that 
      $\LT (\eta I)=\LT (\upsilon I)$ for any $\eta, \upsilon \in U_I$.
\end{theo}

\begin{proof}
   Let $\eta \in U_I. $ We denote $\LT (\eta I_{\le d})$ by $J_{\le d}$. Then $J_{\le d}=\LT (\upsilon I_{\le d})$ for all $\upsilon \in U_I$ and $d>\delta$. Indeed, since $\LT (\eta I_{\le \delta})$ contains the initial terms in the reduced Gr\"obner basis $G$ of $\eta I$,
    \[
    J_{\le d}\subset A_{\le d-\delta}\LT (\eta  I_{\le \delta})=A_{\le d-\delta}\LT ( \upsilon I_{\le \delta})\subset \LT (\upsilon I_{\le d}).
    \]
    As $\dim_K (J_d)=\dim_K (\LT (\upsilon I_d))$, we obtain $J_d=\LT (\upsilon I_d)$ for all $\upsilon \in U_I$. Since $\LT (\eta  I)=\bigcup_{d=\delta}^{\infty} J_{\le d}$, then $\LT (\eta I)=\LT (\upsilon I)$ for any $\eta,\upsilon \in U_I$, which concludes the proof.
\end{proof}

\begin{deftn}
    We call $\LM(\eta I)$, with $\eta\in U_I \subset \GL_n(R)$ as given in Theorem \ref{theo:tropical gin}, the tropical generic initial ideal (tropical gin) of $I$. 
\end{deftn}
Unfortunately, $U_I$ is not a Zariski-open subset of $GL_n(R)$ in general,
hence the \textit{generic} in the name "tropical gin" is only
given as a reference to the classical case.
The following proposition is a consolation.

\begin{prop}
    Assume $k$ is infinite. Then \[U_I \mod \pi:=\{ \eta \mod \pi, \: \text{for } \eta \in U_I \}\] is a non-empty Zariski-open
    set of $GL_n(k).$ 
\end{prop}
\begin{proof}
Let $q$ be the polynomial defining $U_{I_{\le \delta}}$
in the proof of Theorem \ref{theo:tropical gin}.
One can replace $q$ by some $q/\pi^l$ so that
$\overline{q}=q \mod \pi$ is non-zero, and one can check that consequently, since $k$ is infinite,
$U_I \mod \pi = \{ \overline{x} \in \GL_n(k): \: \overline{q}(\overline{x})  \neq 0 \}$
and this is a non-empty Zariski-open set of $GL_n(k)$. 
\end{proof}
%
%
%

\begin{rem} If, \textit{e.g.}, $R=\R \llbracket t \rrbracket$, and one takes $\eta \in GL_n(R)$
at random using a nonatomic distribution over $\R,$ then
$\eta$ belongs to $U_I$ with probability one.
\end{rem}

\subsection{Borel-fixedness}
\label{subsec:borel_fixedness}
In classical cases, a generic initial ideal is Borel-fixed ideal i.e. it is fixed under the action of the Borel subgroup $\mathcal{B}\subset \GL_n(K)$, which is the subgroup of all nonsingular upper triangular matrices. In tropical cases, a generic initial ideal is not always Borel-fixed. However, it can be Borel-fixed under some conditions. 

\begin{ex}
Let $I=(x^2,y^2)$ and $K=\mathbb{Q}_2$ (using $w=[0,0]$ and grevlex). Then in degree two, for a generic change of variables of $x^2 \wedge y^2$ by the matrix $\begin{bmatrix}
a & b \\ 
c & d
\end{bmatrix}$,    we get
in $K[x,y] \bigwedge K[x,y]$:
\[2(a^2bd-ab^2c) x^2 \wedge xy+(a^2 d^2-b^2c^2) x^2 \wedge y^2 +2(acd^2-bc^2d)xy \wedge y^2. \]

Hence the tropical GIN is $x^2 \wedge y^2$ for degree two, and is therefore not Borel-fixed, nor semi-stable for $y$.
\end{ex}

\begin{deftn}
    Let $\mathfrak{B}\subset \GL_n(O_K)$ be the subgroup generated by nonsingular upper triangular matrices whose diagonal entries have valuation $0.$ We call $\mathfrak{B}$ a Borel subgroup. We say that a monomial ideal $J$ is tropical Borel-fixed if $J$ is fixed under the action of  $\mathfrak{B}$. 
\end{deftn}

A direct adaptation of Theorem 4.2.1 and Prop. 4.2.4 of \cite{HH:2011} states that the usual properties of the GIN are preserved, under some conditions. 

\begin{prop}
   Let $d$ be the maximal total degree of the reduced GB of the tropical generic initial ideal of $I$. If  $K=\mathbb{Q}_p$ and  $p\ge d$, or if $\val(\mathbb{Z} \setminus \{0\})=\{ 0 \}$, then the tropical generic initial ideal of $I$ is tropical Borel-fixed and moreover, semi-stable for $x_n.$
\end{prop}


\section{Tropical FGLM}

In this section, we investigate the second part of the FGLM strategy.
Namely, the multiplication matrices of $A/I$ have been 
computed using the algorithms of Section \ref{section:mult_mat}, and 
we can now perform operations in $A/I$ efficiently.

The strategy is then to go through projections in $A/I$
of monomials and find linear relations among them.
When done carefully, these relations provide polynomials
in $I$, whose leading terms for the new term order
 can be read on the monomials
defining the relation.
When processed in the right order, we can obtain from
these polynomials a minimal GB of $I$
 for our new term order.

\subsection{Tropical to classical}

We first begin with the easiest case of
starting from a tropical GB
and computing a classical GB.

It is clear that once the multiplication matrices are obtained,
we can directly apply the classical FGLM algorithm (namely Algorithm 4.1 of \cite{Faugere:1993}, 
see also Algorithm 8 of \cite{Huot:13}), or its $p$-adic stabilized version:
Algorithm 3 of \cite{RV:2016}.
This part is in $O(n \delta^3)$ arithmetic operations.
We refer to Prop 3.6 of \textit{loc. cit.} and 
obtain the following propositions.

\begin{prop}
The total complexity to compute a classical GB of $I$
starting from a tropical GB is in $O(n^3 \delta^3)$
arithmetic operations.
\end{prop}

Behavior regarding to precision can be stated the following way.

\begin{prop}
Let $\leq_1$ be a tropical term ordering and $\leq_2$ be a monomial ordering.
Let $G$ be
  an approximate reduced tropical GB for $\leq_1$ of the ideal $I$, with
  coefficients known up to precision $O(\pi^N)$. 
  Let $\Xi$ be the smallest valuation of a coefficient in $G.$    
Let $B_{\leq_1}$ and $B_{\leq_2}$ be the canonical
  bases of $A/I$ for $\leq_1 $ and $\leq_2$. Let $M$ be the matrix whose
  columns are the $NF_{\leq_1} (x^\beta)$ for $x^\beta \in B_{\leq_2}$. Let $cond_{\leq_1, \leq_2}(I)$ 
  be the biggest valuation of an invariant
  factor  in the Smith Normal Form of $M$. 
  Recall that $D=1+\max_{x^\alpha \in B_\leq} \vert x^\alpha \vert $.
  
  Then if   $N>2cond_{\leq_1, \leq_2}(I)-\left( \frac{(n \delta)^{2D+2}-1}{(n \delta)^2-1} \right) \Xi$,
  we can chain Algorithm \ref{algo:mult_mat} and Algorithm 3 of \cite{RV:2016}
  to obtain an 
  approximate GB $G_2$ 
  of $I$ for $\leq_2$.
  The  coefficients of the polynomials of 
  $G_2$ are known up to precision
  $O\left( \pi^{N+\left( \frac{(n \delta)^{2D+2}-1}{(n \delta)^2-1} \right) \Xi-2cond_{\leq_1, \leq_2}(I)} \right)$.  
\end{prop}

\subsection{Tropical to shape position}
\label{subsec:trop_to_shape}

We can apply any classical FGLM algorithm if $K$ is an exact field,
or a stabilized variant using Smith Normal Form, as in Algorithm 6 of \cite{RV:2016}.
We refer to Prop. 4.5 of \textit{loc. cit.}.
Complexity is very favorable when we have the combination
of Borel-fixedness and shape position.

\begin{prop}
If $I$ is in shape position and semi-stable for $x_n$, then 
we can combine Algorithm \ref{algo:M_n_when_semi_stable} with
Algorithm 6 of \cite{RV:2016}).
The time-complexity is in
  $O(n\delta^2)+O(\delta^3)$ arithmetic operations.
\end{prop}

\begin{prop}
Let $G_1$ be an approximate reduced GB of $I$, with
coefficients known at precision $O(\pi^N)$.
Let $\Xi$ be the smallest valuation of a coefficient in $G_1.$    
If $\leq_2$ is lex, and if we assume that the ideal
  $I$ is in shape position and $LM_{\leq_1}(I)$ is semi-stable for $x_n$,
  then the adapted FGLM in Algorithm 6 of \cite{RV:2016}), computes an
  approximate GB $G_2$ of $I$ for lex, in shape position. The
  coefficients of the polynomials of $G_2$ are known up to precision
  $O(\pi^{N-2cond_{\leq_1, \leq_2}+\delta \Xi})$. 
  Moreover, we can read on $M$ whether the precision was enough
  or not, and hence prove after the computation that the result
  is indeed an approximate GB.
\end{prop}

\subsection{Tropical (or classical) to tropical}

We conclude our series of algorithms with a new algorithm to compute a tropical GB
of $I$ of dimension $0$ knowing  the multiplication matrices
of $A/I.$

In the classical case, the vanilla FGLM algorithm
goes through the monomials $x^\alpha$ in ascending order
for $\leq_2$, test whether $x^\alpha$ is in the 
vector space generated (in $A/I$) by the monomials $x^\beta$ such that
$x^\beta <_{2} x^\alpha$, and if so,
produce a polynomial in the GB in construction
from the relation obtained by this linear relation.

In the tropical case, because of the fact that
coefficients have to be taken into account,
a relation (in $A/I$) between $x^\alpha$ and some monomials
$x^\beta$ such that
$x^\beta <_2 x^\alpha$
is not enough to ensure that $x^\alpha \in \LT_{\leq_2}(I).$
We deal with this issue by \begin{enumerate*}
\item taking all monomials of a given degree at the same time, in a big Macaulay matrix, and
\item reducing them with a special column-reduction algorithm so as to preserve the leading terms.
\end{enumerate*}

The linear algebra algorithm is presented in Algorithm \ref{algo:column_echelon_for_fglm},
with the general tropical FGLM algorithm in Algorithm \ref{algo:trop_fglm}.

\begin{algorithm} 
  \SetKwInOut{Input}{input}\SetKwInOut{Output}{output}
  \Input{$M_1,\dots,M_n$ the multiplication matrices of $A/I$, in a basis $B_{\leq_1}$ for a tropical term ordering $\leq_1$, a tropical term ordering $\leq_2$.}
  \Output{A GB $G$ of the ideal $I$ for $\leq_2$.}

  $L \leftarrow \{1\}$, $G \leftarrow \emptyset$, $d \leftarrow 1$ \;
  $M \leftarrow $ the matrix with $\delta$ rows and $0$ columns \;
  $P \leftarrow $ the matrix with $0$ rows and $0$ columns \;  
  \While{$L \neq \emptyset$}{
    Stack on the right of $M$ all the monomials in $L$ of degree $d$, written in the basis $B_{\leq_1}$ using the multiplication matrices \;
    Remove those monomials from $L$ \;
    Apply Algorithm \ref{algo:column_echelon_for_fglm} with $M$ and $\leq_2$, to get a new $M$ and update the pivoting matrix $P$  \; \tcc{If $M_0$ is the matrix of the $NF_{\leq_1}(x^\alpha)$ for $x^\alpha$ indexing the columns of $M$, then $M=M_0 P.$}
    For all the new columns indexed by $x^\alpha$ that reduced to zero, add to $G$ the polynomial $x^\alpha-\sum_{\gamma \neq \alpha} P_{\gamma, \alpha} x^\gamma$, and remove the multiples of $x^\alpha$ from $L$ \;
    Add to $L$ the $x_i x^\alpha$ for all $i$ and for all $x^\alpha$ new column in $M$ that did not reduce to zero, and remove the duplicates \;
    $d \leftarrow d+1$ \;
  }
  \textbf{Return $G$}
  
  \caption{A tropical FGLM algorithm} \label{algo:trop_fglm}
\end{algorithm}

\begin{algorithm} 
  \SetKwInOut{Input}{input}\SetKwInOut{Output}{output}
  \Input{$M$ a $\delta \times l$ matrix over $K$, whose rows and columns are indexed by monomials. A tropical term ordering $\leq$. An invertible $s \times s$ matrix $P.$ }
  \Output{A column-reduction of $M$ compatible with $\leq$, an updated $P.$ }

  \textbf{if} $M =0$ \textbf{then}	Return $M,P$ \;
  Find the coefficient $M[i,j]$ of row indexed by $x^\beta$ and column indexed by $x^\alpha$ such that $M[i,j]^{-1} x^\alpha$ is smallest, and using smallest $x^\beta$ to break ties \;
  Use this non-zero coefficient to eliminate the other coefficients on the same row  \;
  Update $P$ accordingly \;
  Proceed recursively on the remaining rows and columns \;
  
  \textbf{Return $M,P$}
  
  \caption{Column reduction for FGLM} \label{algo:column_echelon_for_fglm}
\end{algorithm}

The fact that Algorithm \ref{algo:column_echelon_for_fglm} computes a
column-echelon form of the matrix (up to column-swapping) along with the pivoting matrix is clear.
What is left to prove is the compatibility of the pivoting
process with the computation of the normal forms and the leading terms
according to $\leq_2.$
It relies on the following loop-invariant.

\begin{prop}
At any point during the execution of Algorithm \ref{algo:column_echelon_for_fglm}, for any $x^\alpha$, the column of $M$
indexed by $x^\alpha$ corresponds to the normal form  $NF_{\leq_1}(H)$ (with respect to $\leq_1$) of some polynomial $H$ with $LT_{\leq_2}(H)=x^\alpha$. \label{prop:loop_inv_column_red_for_fglm}
\end{prop} \vspace{-.4cm}
\begin{proof}
It is true by construction for any column when entering Algorithm \ref{algo:column_echelon_for_fglm}.
Also by construction, all columns are labelled by distinct monomials.
Now let us assume that on Line 4, we are eliminating a coefficient $d$ on the column labelled 
by $x^\beta$ using a coefficient $c$ on the column labelled by $x^\alpha$ as pivot.
Because of the choice of pivot on Line 3, we get that $c^{-1} x^\alpha <_2 d^{-1} x^\beta.$
Let us assume that the column indexed by $x^\alpha$ corresponds to $NF_{\leq_1}(H)$ with $LT_{\leq_2}(H)=x^\alpha$,
and the column indexed by $x^\beta$ corresponds to $NF_{\leq_1}(Q)$ with $LT_{\leq_2}(Q)=x^\beta$. Please note that $x^\alpha \neq x^\beta.$
Then after pivoting the second column corresponds to $NF_{\leq_1}(Q-dc^{-1}H).$
As $LT_{\leq_2}(dc^{-1}H)=dc^{-1} x^\alpha <_2 x^\beta$, the loop-invariant is then preserved, which is enough to conclude the proof.
\end{proof}

\begin{theo}
Algorithm \ref{algo:trop_fglm} terminates and is correct: its output is a GB of the ideal $I$ for $\leq_2$.
It requires $O(n \delta^3)$ arithmetic operations.
\end{theo} 
\begin{proof}
We use the following loop-invariant: after Line 9 is executed, $LT_{\leq_2}(G)$ contains all the minimal generators
in $LT_{\leq_2}(I)$ of degree $\leq d$, they each correspond to a reduced-to-zero column of $M$, and the $x^\beta$
corresponding to non-reduced-to-zero columns of $M$ are all in $NS_{\leq_2}(I).$ The proof for this invariant is
as follows.
As $\leq_2$ is degree-compatible, it is clear by linear algebra that
$rank(M)=\dim (A_{\leq d} / I_{\leq d}).$
Thanks to Proposition \ref{prop:loop_inv_column_red_for_fglm}, the polynomials added to $G$
are in $I$, and more precisely, $f=x^\alpha-\sum_{\gamma} P_{\gamma, \alpha} x^\gamma$ as in Line 8
is a polynomial such that $\LT_{\leq_2}(f)=x^\alpha$ and $NF_{\leq_1}(f)=0$, as given in the Proposition.
Their $LT_{\leq_2}$'s are minimal generators of $LT_{\leq_2}(I)$ by construction (all multiples of previous generators have
been erased).
By a dimension argument, no minimal generator is missing.

Once $d$ is big enough for all minimal generators of $LT_{\leq_2}(I)$ to have been produced,
no monomials can be left in $L$ and the algorithm terminates.
Termination and correctness are then clear.

As columns are labelled by some $x_i x^\alpha$ with $x^\alpha \in NS_{\leq_2}(I)$
then at most $n \delta$ columns are produced in the algorithm.
As the rank of $M$ is $\delta$ and so is also its number of rows,
the column-reduction of a given column costs $O(\delta^2)$ arithmetic operations.
Consequently, the total cost of the algorithm is in $O(n\delta^3)$ arithmetic operations.
\end{proof}

\begin{rem}
The previous algorithm remarkably bears the
same asymptotic complexity as the vanilla classical FGLM algorithm ($O(n \delta^3)$ arithmetic operations), regardless of the more involved
linear algebra part.
Could fast linear algebra also be applied here?
\end{rem}


\begin{ex}
Let $(x+\frac{1}{2}y,y^2+1)$ be a GB of the ideal it spans,
for $w=[0,-1]$ and grevlex. We compute a GB of the same ideal
for $w=[0,0]$ and grevlex. The following matrices are: the polynomials added to $M$ (in three batches, by degree),
the final state of $M$ and the final $P.$
In the end, we get $(y+2x,x^2+\frac{1}{4})$ as the output GB.

\hspace{-.5cm}
\begin{tikzpicture}
\matrix (m) [matrix of math nodes,nodes in empty cells, ampersand replacement=\& ] at (-1.3,0) {\&  1 \& x \& y \&  x^2   \\ 
 1  \& 1 \& \phantom{-2^{-1}}\& \&-2^{-2} \\
y  \&\phantom{1} \&-2^{-1} \& 1 \&\phantom{ -2^{-2}} \\
} ;
\draw (m-2-2.north west) to[ right] (m-3-2.south west);
\draw (m-2-3.north west) to[ right] (m-3-3.south west);
\draw (m-2-5.north west) to[ right] (m-3-5.south west);
\draw (m-2-5.north east) to[ left] (m-3-5.south east);

\matrix (n) [matrix of math nodes,nodes in empty cells, ampersand replacement=\& ] at (1.7,0) {\&  1 \& x \& y \&  x^2   \\ 
 1  \& 1 \& \&0 \&0 \\
y  \&\phantom{1} \&-2^{-1} \& 0 \&0 \\
} ;
\draw (n-2-2.north west) to[ right] (n-3-2.south west);

\draw (n-2-5.north east) to[ left] (n-3-5.south east);

\matrix (o) [matrix of math nodes,nodes in empty cells, ampersand replacement=\& ] at (4.7,0) {  1 \&  \&  \& 2^{-2} \\ 
    \& 1 \& 2\& \\
    \&  \& 1\& \\
 \phantom{1}   \&  \& \& 1 \phantom{ {}^{-2}}\\
} ;
\draw (o-1-1.north west) to[ right] (o-4-1.south west);
\draw (o-1-4.north east) to[ left] (o-4-4.south east);

\draw (.2,-.5) node{,} ;
\draw (3,-.5) node{,};
\draw (3.5,-.2) node{$P=$} ;
\end{tikzpicture}
\end{ex}
\vspace{-.5cm}
\section{Numerical data}
\label{sec:num_data}

A toy implementation of our algorithms 
in \sage \cite{Sage} 
is available
on \url{https://gist.github.com/TristanVaccon}.
The following arrays gather some numerical results. The timings are expressed in seconds of CPU time.\footnote{Everything was performed on a Ubuntu 16.04
with 2 processors of 2.6GHz and 16 GB of RAM.} 

\subsection{Tropical to classical}

For a given $p,$ we take three polynomials with random coefficients in $\mathbb{Z}_p$
(using the Haar measure)
in $\mathbb{Q}_p[x,y,z]$ of degrees $2 \leq d_1 \leq d_2 \leq d_3 \leq 4.$ $D=d_1+d_2+d_3-2$ is the Macaulay bound.
We first compute a tropical GB
for the weight $w=[0,0,0]$ and the grevlex monomial
ordering, and then apply Algorithms \ref{algo:mult_mat} and \ref{algo:trop_fglm}
to obtain a lex GB. We compare with the strategy of computing a classical
grevlex GB and then applying FGLM to obtain a lex GB.
For any given choice of $d_i$'s, the experiment is repeated 50 times.
Coefficients of the initial polynomials are given at high-enough
precision $O(p^N)$ for no precision issue
to appear (see \cite{RV:2016} for more on FGLM at finite precision).

Coefficients of the output tropical GB or classical GB
are known at individual precision $O(p^{N-m})$ (for some $m \in \mathbb{Z})$).
We compute the total mean and max on those $m$'s on the obtained GB.
In the first following array, we provide the mean and max for the tropical
strategy.
In the second, to compare classical and tropical, we provide couples for the mean on the $50$ ratios of timing per execution ($t$),
along with the arithmetic ($\Sigma$) and geometric ($\pi$) mean of the $50$ ratios of mean loss in
 precision per execution. Data for $p=101$ or $65519$ are not worth for these ratios
as the loss in precision is $0$ most of the time.

In average the tropical strategy takes longer, but save a large amount of precision (for small $p$).
While the ratio of saved precision may decrease with the degree,
the abolute amount of saved precision is often  still very large.
We have also noted that the standard deviations for these ratios can be very large.

\begin{scriptsize}
\hspace{-.5cm}
\begin{tabular}{|l|C{.21cm}|C{.21cm}|C{.21cm}|C{.21cm}|C{.21cm}|C{.21cm}|C{.21cm}|C{.21cm}|C{.21cm}|C{.21cm}|C{.21cm}|C{.21cm}|C{.21cm}|C{.21cm}|}
\hline 
precision (trop.) & \multicolumn{2}{|c|}{ $D=4$} & \multicolumn{2}{|c|}{5} & \multicolumn{2}{|c|}{6}
 & \multicolumn{2}{|c|}{7} & \multicolumn{2}{|c|}{8} & \multicolumn{2}{|c|}{9}  \\ 
\hline 
$p=2$ & 11 &103 &25&278 & 60&509 &176&1253 & 300&1783 & 652&3929 \\ 
\hline 
3 &  3&21  & 12&97  & 36&396  &125&634  &  141&1002 &282 &2876 \\ 
\hline 
101 &  0&1 & 0&1 & 1 & 79  &  0&2  &  15&408  &  0&2   \\ 
\hline
65519 &  0&0  & 0&0 & 0&0 & 0&0 & 0&0  &  0&0   \\ 
\hline  
\end{tabular}

\hspace{-.5cm}
\begin{tabular}{|l|C{.05cm}|C{.05cm}|C{.05cm}|C{.05cm}|C{.05cm}|C{.05cm}|C{.05cm}|C{.05cm}|C{.05cm}|C{.05cm}|C{.05cm}|C{.05cm}|C{.05cm}|C{.05cm}|C{.05cm}|C{.05cm}|C{.05cm}|C{.05cm}|}
\hline 
\multirow{2}{*}{$\frac{\textrm{trop.}}{\textrm{classical}}$} & \multicolumn{3}{|c|}{ $D=4$} & \multicolumn{3}{|c|}{5} & \multicolumn{3}{|c|}{6}
 & \multicolumn{3}{|c|}{7} & \multicolumn{3}{|c|}{8} & \multicolumn{3}{|c|}{9}  \\ 
\cline{2-19} 
 & $t \vphantom{t^2}$ & $\Sigma$ & $\pi$  & $t$ & $\Sigma$ & $\pi$& $t$ & $\Sigma$ & $\pi$& $t$ & $\Sigma$ & $\pi$&  $t$ & $\Sigma$ & $\pi$ &$t$ & $\Sigma$ & $\pi$ \\
\hline 
$p=2$ & 20 & .4 &.3 & 5 & .4 &.2 & 5 & .5 & .2& 5 &  .6 & .2&  1.5 & .8 &.2 & 9 & 1 &.2 \\
\hline 
3 & 6&.6 &.2 & 6&.5 &.2 &5&.5 &.2 & 2&.4 &.1  & 1.2&.7 &.1  &  .9&.9 &.1  \\ 
\hline 
\end{tabular}
\end{scriptsize}

\vspace{-.3cm}
\subsection{Tropical to tropical}

We repeat the same experiments for mean and max loss in precision, but this time
we compute a tropical GB for weight $w=[0,0,0]$
and then use Algorithm \ref{algo:trop_fglm} to compute a
tropical GB for weight $w=[-2,4,-8]$ (grevlex
for tie-breaks in both cases).  
Precision-wise, it seems that there is an intrinsic difficulty
in computing a lex GB compared to a tropical GB.


\begin{footnotesize}
\hspace{-.5cm}
\begin{tabular}{|l|C{.15cm}|C{.15cm}|C{.15cm}|C{.15cm}|C{.15cm}|C{.15cm}|C{.15cm}|C{.15cm}|C{.15cm}|C{.15cm}|C{.15cm}|C{.15cm}|C{.15cm}|C{.15cm}|}
\hline 
precision loss & \multicolumn{2}{|c|}{ $D=4$} & \multicolumn{2}{|c|}{5} & \multicolumn{2}{|c|}{6}
 & \multicolumn{2}{|c|}{7} & \multicolumn{2}{|c|}{8} & \multicolumn{2}{|c|}{9}  \\ 
\hline 
$p=2$ &  2&18 &2.5&14 & 2.6&14 &2.9&16 & 3&17 & 3.5&19 \\ 
\hline 
3 &  1&9  & 1&7  & 1&9  &1.4 &14&1.4 &  11&2 &13 \\ 
\hline 
101 &  0&1 & 0&1 &  0& 1  &  0&2  &  0&2 &  0&2   \\ 
\hline
65519 &  0&0  & 0&0 & 0&0  & 0&0 & 0&0  &  0&0   \\ 
\hline  
\end{tabular}
\end{footnotesize}


\vspace{-.3cm}
\subsection{Semi-stability and shape position}

We adapt our setting to $\mathbb{Q} ((t))$,
using entries with coefficients in $\mathbb{Z} \llbracket t \rrbracket$ given at precision 50
(using \sage 's built-in random function), 
and apply the ideas of Subsection \ref{subsec:using_semi_stability} and Section \ref{sec:gin}.
As $\mathbb{Q}$ is involved, computations are slow for $D \geq 7$ due
to coefficients growth.

\begin{footnotesize}

\hspace{-.5cm}
\begin{tabular}{|l|c|c|c|c|c|c|}
\hline 
$w=[0,0,0]+$grevlex & \multicolumn{2}{|c|}{ $D=4$} & \multicolumn{2}{|c|}{5} & \multicolumn{2}{|c|}{6}
\\ 
\hline 
mean timing (F5 $\&$ FGLM) &  2.8&9.4 &3.9&102 & 10&1030  \\ 
\hline 
precision F5 (mean $\&$ max) &  0&2  & 0&2  & 0&3  \\ 
\hline 
precision FGLM (mean $\&$ max)& 0 &  0&0.1 & 8&0.4 &  34 \\ 
\hline
\end{tabular}
\end{footnotesize}

\begin{small}
\bibliographystyle{plain}

\end{small}

\end{document}